\documentclass[11pt]{article}

\usepackage{amsmath,amssymb,amsthm}
\usepackage{algorithm}
\usepackage{algorithmic}
\usepackage{geometry}
\usepackage{tikz}
\usetikzlibrary{calc}
\usepackage{hyperref}

\newtheorem{theorem}{Theorem}
\newtheorem{lemma}[theorem]{Lemma}

\newtheorem*{remark}{Remark}

\title{A Simple Algorithm for Clustering Discrete Distributions}
\author{Pradipta Mitra}
\date{}

\begin{document}

\maketitle

\paragraph{Preface.}
This article is excerpted from Chapter~5 of the author's PhD thesis \cite{mitra08thesis}.
The decision to make it available as a standalone article was prompted by the recent discovery
that the content of this chapter has been cited a few times \cite{neumann18, seifchen24}.

\begin{abstract}
  We propose a simple, projection-based algorithm for clustering mixtures of discrete (Bernoulli) distributions. Unlike previous approaches that rely on coordinate-specific ``combinatorial
  projections,'' our algorithm is rotationally invariant and works by projecting samples onto approximate centers obtained via a $k$-means computation on the best rank-$k$ approximation of
  the data matrix. This resolves a conjecture of McSherry on the existence of such geometric algorithms for discrete distributions. The same algorithm also applies to continuous distributions
   such as high-dimensional Gaussians, providing a unified approach across distribution types. We prove that the algorithm succeeds under a natural separation condition on the cluster
  centers.
  \end{abstract}

\section{Introduction}

In this paper, we propose a projection-based, rotationally invariant algorithm for
clustering mixture of distributions, including discrete distributions. This resolves a
conjecture of McSherry \cite{mcsherry01}. It also gives first clustering algorithms for a class of discrete
distributions previously not dealt with, and gives a natural algorithm that works for
wide class of discrete and continuous distributions.

We are concerned here with probabilistic mixture models. We start with $k$ ``simple'' probability distributions $D_1, D_2, \ldots, D_k$ where each is a distribution on $\mathbb{R}^n$. With
each $D_r$ we can associate a center $\mu_r$ of the distribution, defined by $\mu_r = E_{D_r}(v)$. A
mixture of these probability distributions is a density of the form $w_1 D_1 + w_2 D_2 +
\ldots + w_k D_k$, where the $w_i$ are fixed non-negative reals summing to 1. A clustering
question in this case would then be: What is the required condition on the probability distributions so that given a collection of vectors generated from the mixture, we
can group the vectors into $k$ clusters, where each cluster consists precisely of vectors
picked according to one of the component distributions of the mixture?

The ``planted partition model'' we have been studying so far is a special case of the
mixture model, except for the added requirement of symmetry. Symmetry only
adds very mild dependence to the model, and can be handled easily. Consequently,
we will restrict ourselves to mixture models. We present and analyze a very natural
algorithm that works for quite general mixture models. The algorithm is completely
``geometric'', in the sense that it is not coordinate specific in any way, is based
on ordinary projections to natural subspaces, and is robust under rotation. Our
algorithm is the first such algorithm shown to work for discrete distributions. The
existence of such an algorithm was first conjectured by McSherry in \cite{mcsherry01}.

Why is this interesting? To elucidate some of these reasons, we need to introduce
the common techniques used in the literature and their limitations.

The method very commonly used to cluster mixture models (discrete and continuous) is ``spectral clustering'' --- which refers to a large class of algorithms where
the eigenvectors or singular vectors of $A$ are used in some fashion. Many works have
come out on spectral clustering for continuous distributions, which have focused on
high-dimensional gaussians and their generalizations \cite{achlioptas05, dasgupta07, kannan05spectral, vempala04}. On the other hand, the use of eigenvectors for discrete, graph-based models was pioneered by Boppana \cite{boppana87}. The essential aspects of many techniques used now first appeared in \cite{alon98, alon97}. These results were generalized by McSherry \cite{mcsherry01} and a large body of work followed (\cite{dasgupta07, dasgupta04} etc).

Despite their variety, these algorithms have a common first step. Starting with
well-known spectral norm bounds on $\|A - E(A)\|$ (e.g.\ \cite{vu05}, there are many others),
a by-now standard analysis can be used to show that the projection on best $k$-dimensional approximation would find a good approximate partitioning of the data.
The major effort, often, goes into ``cleaning up'' this approximate partitioning.

At this point the methods for discrete and continuous methods diverge. For
continuous models, spectral clustering is often followed by projecting new samples
on the spectral subspace obtained. Here is why, broadly speaking, this works: the
spectral norm bounds imply that the spectral subspace is a span of the real centers
of the data, plus some error. This error can be thought of as a (or a small number of)
vector(s) of small $\ell_2$ norm. Now, if we consider a new, random sample, this sample
will be almost orthogonal to the error vector(s) with high probability. Hence, it is
as if we were projecting to the span of centers, which is a good subspace to project
to, and the clustering is complete.

Unfortunately, this is simply not true for discrete distributions. This has resulted
in rather ad-hoc methods for cleaning up mixture of discrete distributions. A most
pertinent example would be ``combinatorial projections'', proposed by McSherry \cite{mcsherry01},
and all other algorithms for discrete distributions share the same central features
we are concerned with. Here, one uses the span of characteristic vectors of the
approximate clusters as the subspace to project down to. Though successful in
that particular context, the idea of combinatorial projections is unsatisfactory for a
number of inter-related reasons. First, such algorithms aren't robust under many
natural transformations of the data for which we expect the clustering to remain the
same. The most important case is perhaps rotation. It is natural to expect that if
all vectors of the data are rotated by the same rotation matrix, the clustering would
remain the same. Projective algorithms would continue to work in these cases, while
combinatorial methods might not. To take the case of ``combinatorial projection'',
the behavior of the algorithm becomes unpredictable if the vectors are rotated. A
related issue is this: the overriding idea used to cluster discrete distributions is that
the \emph{feature space} (in addition to the object space) has clearly delineated clusters.
One way to state this condition is that $E(A)$ can be divided into a small number
of sub matrices, each of which contains the same value for each entry. This is
implicit in graph-based models, as the objects and features are the same things ---
vertices. These results will extend to the case of rectangular matrices as long as some
generalization of that condition holds, but for general cases where the centers don't
necessarily have such strong structure it is not clear how to extend these ideas. Our
algorithm solves this problem.

We will focus on discrete distributions (i.e.\ Bernoulli distributions) where each
entry of a sample vector is an independent $0/1$ variable. Discrete distributions are the
``hard'' distributions in the present context. That our algorithm will work for many
usual continuous distributions, such as spherical or sphere-like high-dimensional
gaussians will not be very hard to see.

This paper is organized as follows. In Section~\ref{sec:model} we present the basic model
and assumptions, our results and review of more relevant literature. Next, in Section~\ref{sec:algorithm} the algorithm is presented. Section~\ref{sec:analysis} consists of the technical results that prove
our bounds.

\section{Model}\label{sec:model}

There are $k$-centers $\mu_r$, such that $0 \le \mu_r(i) \le \sigma^2 \le 1$ for all $r \in [k]$, $i \in [n]$, and some
$\sigma^2 \ge \frac{\log^6 n}{n}$. Each center defines a probability distribution on $\mathbb{R}^n$: an $n$-dimensional
sample $v$ from this distribution is generated by setting $v(i) = 1$ with probability
$\mu_r(i)$ and $0$ otherwise, independently for all $i \in [n]$.

With each distribution we associate weight $w_r \ge 0$, such that $\sum_{r \in [k]} w_r = 1$. The
data matrix is generated in the following way. For each distribution, a set $T_r$ of $w_r m$
samples are chosen from it, adding up to $\sum_{r \in [k]} w_r m = m$ total samples. The $m$
samples are arranged as columns of a $n \times m$ matrix $A$, which is presented as the
data. We will use $A$ to mean both the matrix, and the set of vectors that are rows of
$A$, where the particular usage will be clear from the context. $E(A)$ is defined by rule:
$A_i \in T_r$ then $(E(A))_i = \mu_r$ ($M_i$ is the $i$th column of $M$). The algorithmic problem is:
Given $A$, can we find the clusters $T_1, T_2, \ldots, T_k$?

For successful clustering, we need to assume the following:

\medskip
\noindent\textbf{Separation condition:} We assume, for all $r, s \in [k]$, $r \ne s$
\begin{equation}\label{eq:separation}
\|\mu_r - \mu_s\|^2 \ge 8100 c k \sigma^2 \frac{1}{w_{\min}} \left(1 + \frac{n}{m} + \log m\right)
\end{equation}
for some constant $c$.

The following can be proved using large deviation inequalities (for example, Theorem~\ref{thm:bernstein}). For a sample $v \in T_r$, with probability $1 - \frac{1}{m^4}$,
\begin{equation}\label{eq:deviation}
|(v - \mu_r) \cdot (\mu_s - \mu_r)| \le \frac{1}{10} \|\mu_s - \mu_r\|^2
\end{equation}
This is a quantitative version of the property that the distance of a sample to its
own center is smaller than its distance to the center of another cluster, but this
representation has some technical advantages. In fact, one can reasonably consider
this as an assumption of the model, since without such a property, it's not clear how
one could cluster the data even in principle.

\begin{figure}[ht]
\centering
\begin{tikzpicture}[scale=1.5]
  \coordinate (mu_r) at (0,0);
  \coordinate (mu_s) at (4,0);
  \coordinate (v) at (1.5,2);
  \draw[->,thick] (mu_r) -- (mu_s) node[right]{$\mu_s$};
  \draw[->,thick] (mu_r) -- (v) node[above]{$v$};
  \draw[dashed] (v) -- ($(mu_r)!(v)!(mu_s)$);
  \node[below left] at (mu_r) {$\mu_r$};
\end{tikzpicture}
\caption{If $v \in T_r$, the projection of $v - \mu_r$ on $\mu_s - \mu_r$ will be closer to $\mu_r$ than it is to $\mu_s$.}
\label{fig:projection}
\end{figure}
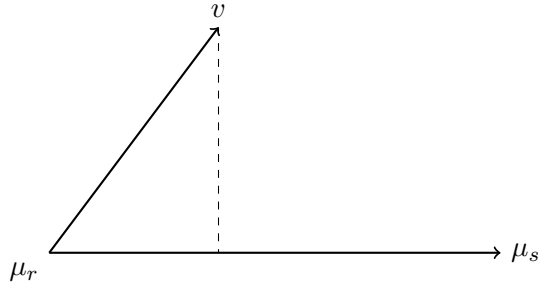

\paragraph{Related Work.}
Let us discuss another example to further illustrate the dichotomy between discrete
and continuous distributions in the literature. Two related papers, \cite{dasgupta04} and \cite{achlioptas05}
employ a very natural linkage algorithm to perform the clean-up phase, and share
much of their algorithm. However, these papers respectively address discrete and
continuous distributions, and the difference influences the algorithm used in telling
ways. In \cite{achlioptas05}, the authors generalize high-dimensional Gaussians by adopting the
notion of $f$-concentrated distributions. Consequently, a simple projection algorithm
suffices. In \cite{dasgupta04}, on the other hand, the clean-up procedure uses the ``combinatorial
projection'' method pioneered in \cite{mcsherry01}.

\section{The Algorithm}\label{sec:algorithm}

The algorithm, reduced to its essentials, is simple and natural. First, we randomly
divide the rows into two equal parts --- $A_1$ and $A_2$. We will use information of one
part to cluster the other part, as this allows to utilize independence of these parts.
For $A_1$, we will find its best rank $k$ approximation $A_1^{(k)}$ (by computing the Singular
Value Decomposition of the matrix). We will find approximately correct centers by
a procedure described in the next few paragraphs. Once we have the centers, we
measure the distance of each sample in $A_2$ (the other part) to these centers, and
group the sample with the closest center.

To find the approximate centers we will have to solve the solution to the $\ell_2^2$-clustering problem (also known as the ``$k$-means problem'').

\medskip
\noindent $\ell_2^2$ \textbf{clustering problem:} Given a set of vectors $S = \{v_1, \ldots, v_l\}$ in $\mathbb{R}^d$ and a
positive integer $k$, the problem is to find $k$ points $f_1, \ldots, f_k \in \mathbb{R}^d$ (called ``centers'') so
as to minimize the sum of squared distances from each vector $v_i$ to its closest center.
This naturally defines a partitioning of the $n$ points into $k$ clusters. We compute
average of the vectors in each cluster, and these will be the approximate centers.
Quite efficient constant factor deterministic approximation algorithms (even PTAS's)
are available for this problem (See \cite{drineas99, jain01, kanungo04} and references therein). We will work
with the factor 9 approximation algorithm of Kanungo et.\ al.\ \cite{kanungo04}.

\begin{algorithm}
\caption{$\textsc{Cluster}(A, k)$}
\begin{algorithmic}[1]
\STATE Randomly divide the rows $A$ into two equal parts $A_1$ and $A_2$
\STATE $(\theta_1, \ldots, \theta_k) = \textsc{Centers}(A_1, k)$
\STATE $(\nu_1, \ldots, \nu_k) = \textsc{Centers}(A_2, k)$
\STATE $(P_1^1, \ldots, P_k^1) = \textsc{Project}(A_1, \nu_1, \ldots, \nu_k)$
\STATE $(P_1^2, \ldots, P_k^2) = \textsc{Project}(A_2, \theta_1, \ldots, \theta_k)$
\STATE Assume $P_r^1$ and $P_r^2$ correspond the same clusters for all $r$ (this can be done by comparing centers and permuting the index set accordingly)
\STATE \textbf{return} $(P_1^1 \cup P_1^2, \ldots, P_k^1 \cup P_k^2)$
\end{algorithmic}
\end{algorithm}

\begin{algorithm}
\caption{$\textsc{Centers}(A, k)$}
\begin{algorithmic}[1]
\STATE Find $A^{(k)}$, the best rank-$k$ approximation of the matrix $A$.
\STATE Solve the $\ell_2^2$ clustering problem with the following parameters
\begin{itemize}
\item The input vectors are rows of $A^{(k)}$
\item $k$ is the number of centers we seek
\end{itemize}
\STATE Let the clusters computed by the $\ell_2^2$ algorithm be $P_1 \ldots P_k$
\STATE For all $r$, compute $\mu_r^* = \frac{1}{|P_r|} \sum_{v \in P_r} v$
\STATE \textbf{return} $(\mu_1^*, \ldots, \mu_k^*)$
\end{algorithmic}
\end{algorithm}

\begin{algorithm}
\caption{$\textsc{Project}(A, \mu_1^*, \ldots, \mu_k^*)$}
\begin{algorithmic}[1]
\FOR{each $v \in A$}
  \FOR{each $r \in [k]$}
    \IF{$|(v - \mu_r^*) \cdot (\mu_s^* - \mu_r^*)| \le |(v - \mu_s^*) \cdot (\mu_r^* - \mu_s^*)|$ for all $r \ne s$}
      \STATE Put $v$ in $P_r$ (breaking ties arbitrarily)
    \ENDIF
  \ENDFOR
\ENDFOR
\STATE \textbf{return} $(P_1, \ldots, P_k)$
\end{algorithmic}
\end{algorithm}

\section{Analysis}\label{sec:analysis}

The following theorem of Vu \cite{vu05} bounds the spectral norm of $\|A - E(A)\|$:

\begin{theorem}\label{thm:vu}
If $A$ is generated as described in Section~\ref{sec:model}, then with probability $1 - o(1)$
\begin{equation}\label{eq:spectralnorm}
\|A - E(A)\|^2 \le c\sigma^2(m + n)
\end{equation}
for some (small) constant $c$.
\end{theorem}

The following was proved by McSherry \cite{mcsherry01}:

\begin{lemma}\label{lem:mcsherry}
With probability $1 - o(1)$, the best rank-$k$ approximation matrix $A^{(k)}$ satisfies
\[
\|A^{(k)} - E(A)\|_F^2 \le 8ck\sigma^2(m+n)
\]
\end{lemma}

\begin{proof}
As both $A^{(k)}$ and $E(A)$ are of rank $k$, $A^{(k)} - E(A)$ has rank at most $2k$.
Then
\begin{align*}
\|A^{(k)} - E(A)\|_F^2 &\le 2k\|A^{(k)} - E(A)\|^2 \\
&= 2k\|A^{(k)} - A + A - E(A)\|^2 \\
&\le 4k(\|A - A^{(k)}\|^2 + \|A - E(A)\|^2) \\
&\le 8k(\|A - E(A)\|^2) \le 8ck\sigma^2(m+n) \qedhere
\end{align*}
\end{proof}

\begin{lemma}\label{lem:Bs}
Given an instance of the clustering problem, consider any vector $y$ such
that, for some $r$
\[
\|y - \mu_r\|^2 \le \frac{1}{9}\|\mu_r - \mu_s\|^2
\]
for all $s \ne r$. Let $B_s = \{x \in T_s : \|A^{(k)}(x) - y\|^2 \le \frac{1}{4}\|\mu_r - \mu_s\|^2\}$. Also let
$\sum_{B_s} \|A^{(k)}(x) - \mu_s\|^2 = E_s$. Then for all $s \ne r$,
\[
|B_s| \le \frac{40 E_s}{\|\mu_r - \mu_s\|^2}
\]
And specifically,
\[
|B_s| \le \frac{320ck\sigma^2(m+n)}{\|\mu_r - \mu_s\|^2}
\]
\end{lemma}

\begin{proof}
From the conditions of the Lemma,
\begin{align*}
\|y - \mu_s\| &= \|y - \mu_r + \mu_r - \mu_s\| \\
&\ge \|\mu_r - \mu_s\| - \|y - \mu_r\| \\
&\ge 0.66\|\mu_r - \mu_s\|
\end{align*}

Now, for each $x \in B_s$
\begin{align*}
\|A^{(k)}(x) - \mu_s\| &\ge (\|y - \mu_s\| - \|A^{(k)}(x) - y\|) \\
&\ge 0.16 \times \|\mu_r - \mu_s\|
\end{align*}

By assumption,
\[
\sum_{B_s} \|A^{(k)}(x) - \mu_s\|^2 = E_s
\]
\[
\Rightarrow \quad 0.0256 \times \|\mu_r - \mu_s\|^2 |B_s| \le E_s
\]
\[
\Rightarrow \quad |B_s| \le \frac{40 E_s}{\|\mu_r - \mu_s\|^2}
\]

Now as
\[
E_s = \sum_{B_s} \|A^{(k)}(x) - \mu_s\|^2 \le \|A - E(A)\|_F^2 \le 8ck\sigma^2(m+n)
\]
\[
|B_s| \le \frac{320ck\sigma^2(m+n)}{\|\mu_r - \mu_s\|^2}
\]
This proves the Lemma.
\end{proof}

The important aspect of the previous Lemma is that the size of $|B_s|$ goes down as
$\|\mu_r - \mu_s\|^2$ increases. This property will be used later.

We characterize the clustering produced in $\textsc{Centers}(A, k)$ in the following Lemma.

\begin{lemma}\label{lem:centers_clustering}
Consider a clustering $P_1, P_2, \ldots, P_k$ produced by the algorithm $\textsc{Centers}(A, k)$.
We claim:
\begin{itemize}
\item Each $P_r$ can be identified with a unique and different $T_s$ such that
\begin{equation}\label{eq:PrTs}
|P_r \cap T_s| \ge \frac{4}{5} m_s
\end{equation}
Without loss of generality, we shall assume $s = r$.
\end{itemize}

We define
\begin{align*}
\bar{P}_r &= T_r \cap P_r \\
Q_r' &= P_r - \bar{P}_r \\
Q_{rs}' &= P_s \cap Q_r' \\
E_{rs} &= \sum_{v \in Q_{rs}'} \|A^{(k)}(v) - \mu_r\|^2
\end{align*}

Then for all $r, s \in [k]$ such that $r \ne s$
\begin{equation}\label{eq:qrs}
q_{rs}' = |Q_{rs}'| \le \frac{40 E_{rs}}{\|\mu_r - \mu_s\|^2} \le \frac{320ck\sigma^2(m+n)}{\|\mu_r - \mu_s\|^2}
\end{equation}
\end{lemma}

\begin{proof}
First we claim that there is a solution to the $\ell_2^2$ clustering problem so
that the cost of the solution $\le 8ck\sigma^2(m+n)$. Set the centers to be $f_r = \mu_r$ for all $r$
and let the cost of this solution be $C$.
Now,
\begin{align*}
C &\le \sum_r \sum_{A^{(k)}(i) \in [T_r]} \|A^{(k)}(i) - \mu_r\|^2 \\
&= \|A^{(k)} - E(A)\|_F^2 \\
&\le 8ck\sigma^2(m+n)
\end{align*}
By Lemma~\ref{lem:mcsherry}.

\begin{remark}
The problem calls for centers $f_r$ to be in the span of vectors in $A^{(k)}$,
and $\mu_r$ might not be. But this simply strengthens the result.
\end{remark}

Accordingly, $\textsc{Cluster}(G, k)$ produces a solution to the $\ell_2^2$ clustering problem with
cost no more than $72ck\sigma^2(m+n)$ (since we compute a factor 9 approximation to the
solution). We claim, for each $P_r$ the center $f_r$ will be such that, for all $s \ne r$
\begin{equation}\label{eq:fr_close}
\|\mu_r - f_r\|^2 \le \frac{1}{9}\|\mu_r - \mu_s\|^2
\end{equation}

If this is not true for some $r$, then the cost of the solution is at least
\begin{align*}
&\sum_{A^{(k)}(i) \in T_r} \|A^{(k)}(i) - f_r\|^2 \\
\ge\;& \sum_{A^{(k)}(i) \in T_r} \|A^{(k)}(i) - \mu_r + f_r - \mu_r\|^2 \\
\ge\;& \frac{1}{4} \sum_{A^{(k)}(i) \in T_r} \|f_r - \mu_r\|^2 - 3 \sum_{A^{(k)}(i) \in T_r} \|A^{(k)}(i) - \mu_r\|^2 \\
\ge\;& \frac{1}{36 w_{\min}} 8100ck\sigma^2 \left(1 + \frac{n}{m}\right) \log m \cdot m_r - 216ck\sigma^2(m+n) \\
\ge\;& 9ck\sigma^2(m+n)
\end{align*}
which is a contradiction. We used the bound $\|u+v\|^2 \ge \frac{1}{4}\|u\|^2 - 3\|v\|^2$ for the
second inequality, and the assumption that Eqn~\eqref{eq:fr_close} is false for the third inequality.
Assuming \eqref{eq:fr_close}, Lemma~\ref{lem:Bs} implies \eqref{eq:qrs}. Equation~\eqref{eq:PrTs} follows from essentially the
same argument.
\end{proof}

We now focus on the centers $\mu_r^*$ produced by $\textsc{Centers}(A, k)$. First we would like
to show that they are close to the real centers.

\begin{lemma}\label{lem:centers_close}
Let $\mu_1^*, \ldots, \mu_k^*$ be the centers returned by the procedure $\textsc{Centers}(A, k)$.
Then for all $r \in [k]$, and for all $s \ne r$
\begin{equation}\label{eq:centers_close}
\|\mu_r^* - \mu_r\|^2 \le 81ck\sigma^2 \frac{1}{w_{\min}} \left(1 + \frac{n}{m}\right) \le \frac{1}{20}\|\mu_r - \mu_s\|^2
\end{equation}
\end{lemma}

\begin{proof}
We know,
\[
\mu_r^* = \frac{1}{p_r} \sum_{v \in P_r} v
\]
\[
\Rightarrow \quad p_r \mu_r^* = \sum_{v \in P_r} v = \sum_{\bar{P}_r} v + \sum_s \sum_{Q_{rs}'} v
\]
\[
\Rightarrow \quad p_r(\mu_r^* - \mu_r) = \sum_{\bar{P}_r} (v - \mu_r) + \sum_s \sum_{Q_{rs}'} (v - \mu_r)
\]

Then,
\begin{equation}\label{eq:pr_bound}
\|p_r(\mu_r^* - \mu_r)\| \le \left\|\sum_{\bar{P}_r} (v - \mu_r)\right\| + \sum_s \left\|\sum_{Q_{rs}'} (v - \mu_r)\right\|
\end{equation}

Let the samples in $\bar{P}_r$ be $v^1, \ldots, v^{\bar{p}_r}$. Let's define the matrix $S$ with these samples
as columns:
\[
S = \begin{pmatrix} v^1 & v^2 & \cdots & v^{\bar{p}_r} \end{pmatrix}
\]
And $U$ be the matrix of same dimensions with all columns equal to $\mu_r$
\[
U = \begin{pmatrix} \mu_r & \mu_r & \cdots & \mu_r \end{pmatrix}
\]
Also let $\mathbf{1} = \{1, \ldots, 1\}^T$, with $\bar{p}_r$ entries. Now, by Theorem~\ref{thm:vu},
\begin{align*}
\|S - U\| &\le \|A - E(A)\| \le \sigma\sqrt{c(m+n)} \\
\Rightarrow \quad \|(S - U)\mathbf{1}\| &\le \|\mathbf{1}\|\sigma\sqrt{c(m+n)} \le \sigma\sqrt{c(m+n)\bar{p}_r}
\end{align*}

As
\[
\sum_{\bar{P}_r} (v - \mu_r) = (S - U)\mathbf{1}
\]
\begin{equation}\label{eq:Pr_bar_bound}
\left\|\sum_{\bar{P}_r} (v - \mu_r)\right\| \le c\sigma\sqrt{(m+n)\bar{p}_r}
\end{equation}

Now, for any $s$
\[
\left\|\sum_{Q_{rs}'} (v - \mu_r)\right\| \le \left\|\sum_{Q_{rs}'} (v - \mu_s)\right\| + \|q_{rs}'\mu_s - q_{rs}'\mu_r\|
\]

But we know by Lemma~\ref{lem:Bs}
\[
q_{rs}' \le \frac{40 E_{rs}}{\|\mu_r - \mu_s\|^2}
\]

Then,
\begin{align}
\|q_{rs}'\mu_s - q_{rs}'\mu_r\| &= ((q_{rs}')^2 \|\mu_s - \mu_r\|^2)^{\frac{1}{2}} \le (40 q_{rs}' E_{rs})^{\frac{1}{2}} \notag \\
\Rightarrow \quad \|q_{rs}'\mu_s - q_{rs}'\mu_r\| &\le \sqrt{q_{rs}' E_{rs}} \label{eq:qrs_mu_bound}
\end{align}

On the other hand, through an argument similar to the bound for $\sum_{\bar{P}_r} (v - \mu_r)$
\begin{equation}\label{eq:Qrs_bound}
\left\|\sum_{Q_{rs}'} (v - \mu_s)\right\| \le \sigma\sqrt{cq_{rs}'(m+n)}
\end{equation}

Combining equations \eqref{eq:Pr_bar_bound}--\eqref{eq:Qrs_bound},
\begin{align*}
\|p_r(\mu_r^* - \mu_r)\| &\le \sigma\sqrt{c(m+n)\bar{p}_r} + \sum_s \sigma\sqrt{cq_{rs}'(m+n)} + \sum_s \sqrt{40 q_{rs}' E_{rs}} \\
&\le \sigma\sqrt{ck(m+n)p_r} + \sum_s \sqrt{40 q_{rs}' E_{rs}} \\
&\le 2\sigma\sqrt{ck(m+n)p_r} + \sqrt{\sum_s 40 q_{rs}'} \cdot \sqrt{\sum_s E_{rs}}
\end{align*}
Using the Cauchy--Schwartz inequality a few times. But we know that $\sum_s E_{rs} \le 8ck\sigma^2(m+n)$ and $\sum q_{rs}' \le \frac{1}{5}p_r$ (Eqn~\eqref{eq:PrTs} in Lemma~\ref{lem:centers_clustering}). Hence,
\begin{align*}
\|p_r(\mu_r^* - \mu_r)\| &\le \sigma\sqrt{ck(m+n)p_r} + 8\sqrt{ckp_r\sigma^2(m+n)} \\
&\le 9\sigma\sqrt{ck(m+n)p_r}
\end{align*}
\[
\Rightarrow \quad \|\mu_r^* - \mu_r\| \le 9\sigma\sqrt{ck \frac{(m+n)}{p_r}} \le 9\sigma\sqrt{\frac{ck}{w_{\min}}\left(1 + \frac{n}{m}\right)} \qedhere
\]
\end{proof}

Finally we would like to show that $\textsc{Project}(A, \mu_1^*, \ldots, \mu_k^*)$ returns an accurate partitioning of the data. To this end, we claim that $(v - \mu_r^*) \cdot (\mu_r^* - \mu_t^*)$ behaves essentially like $(v - \mu_r) \cdot (\mu_r - \mu_t)$.

\begin{lemma}\label{lem:project}
For each sample $u$, if $u \in T_r$, then for all $s \ne r$
\begin{equation}\label{eq:project_bound}
|(u - \mu_r^*) \cdot (\mu_r^* - \mu_t^*)| \le \frac{2}{5}\|\mu_r - \mu_t\|^2
\end{equation}
with high probability.
\end{lemma}

\begin{proof}
Assume $\mu_r^* = \mu_r + \delta_r$; $\forall r$. Then,
\begin{align*}
(u - \mu_r^*) \cdot (\mu_t^* - \mu_r^*) &= (u - \mu_r - \delta_r) \cdot (\mu_t - \mu_r - \delta_r + \delta_t) \\
&= (u - \mu_r) \cdot (\mu_r - \mu_t) - \delta_r \cdot (\mu_t - \mu_r) - \delta_r \cdot (\delta_t - \delta_r) + (u - \mu_r) \cdot (\delta_t - \delta_r)
\end{align*}

Let's consider each term in the last sum separately.
By equation~\eqref{eq:deviation},
\begin{equation}\label{eq:term1}
(u - \mu_r) \cdot (\mu_r - \mu_t) \le \frac{1}{10}\|\mu_r - \mu_t\|^2
\end{equation}

We have already shown (Lemma~\ref{lem:centers_close}) that $\|\delta_r\| \le \frac{1}{\sqrt{20}}\|\mu_r - \mu_t\|$. Then
\begin{equation}\label{eq:term2}
|\delta_r \cdot (\mu_t - \mu_r)| \le \|\delta_r\|\|\mu_r - \mu_t\| \le \frac{1}{\sqrt{20}}\|\mu_r - \mu_t\|^2
\end{equation}
And
\begin{equation}\label{eq:term3}
|\delta_r \cdot (\delta_t - \delta_r)| \le \frac{1}{20}\|\mu_r - \mu_t\|^2
\end{equation}

The remaining term is $(u - \mu_r) \cdot (\delta_r + \delta_t)$. We prove in Lemma~\ref{lem:bernstein_app} (below) that
\begin{equation}\label{eq:term4}
|(u - \mu_r) \cdot (\delta_r + \delta_t)| \le \frac{1}{100}\|\mu_r - \mu_s\|^2
\end{equation}

Combining equations \eqref{eq:term1}--\eqref{eq:term4}, we get the proof.
\end{proof}

For the proof of Lemma~\ref{lem:bernstein_app}, we will need Bernstein's inequality (see \cite{petrov75} for a reference):

\begin{theorem}[Bernstein's inequality]\label{thm:bernstein}
Let $\{X_i\}_{i=1}^n$ be a collection of independent, almost surely absolutely
bounded random variables, that is $\Pr\{|X_i| \le M\} = 1$ $\forall i$. Then, for any $\varepsilon \ge 0$
\begin{equation}\label{eq:bernstein}
\Pr\left\{\sum_{i=1}^n (X_i - E[X_i]) \ge \varepsilon\right\} \le \exp\left(-\frac{\varepsilon^2}{2\left(\theta^2 + \frac{M}{3}\varepsilon\right)}\right)
\end{equation}
where $\theta = \sum E X_i^2$.
\end{theorem}

In the following, it is crucial to assume that the sample $u$ is independent of $\delta_r$
and $\delta_t$. We can assume this because we use centers from $A_1$ on samples from $A_2$,
and vice-versa, which are independent.

\begin{lemma}\label{lem:bernstein_app}
If $u \in D_r$ is a sample independent of $\delta_r$ and $\delta_t$, for all $t \ne r$
\begin{equation}\label{eq:bernstein_app}
|(u - \mu_r) \cdot (\delta_r + \delta_t)| \le 15ck\sigma^2 \frac{1}{w_{\min}} \left(1 + \frac{n}{m} + \log m\right) \le \frac{1}{100}\|\mu_r - \mu_s\|^2
\end{equation}
with high probability.
\end{lemma}

\begin{proof}
It suffices to prove a bound on $(u - \mu_r) \cdot \delta_r$. The case for $\delta_t$ is similar.
As
\begin{align*}
(u - \mu_r) \cdot \delta_r &= \sum_{i \in [n]} (u(i) - \mu_r(i))\delta_r(i) \\
&= \sum_{i \in [n]} x(i)
\end{align*}
where $x(i) = (u(i) - \mu_r(i))\delta_r(i)$. This is a sum of independent random variables.
Note that by Lemma~\ref{lem:centers_close}
\[
\|\delta_r\|^2 \le 81ck\sigma^2 \frac{1}{w_{\min}} \left(1 + \frac{n}{m}\right)
\]

Now,
\begin{align*}
E(x(i)) &= E(u(i) - \mu_r(i))\delta_r(i) = 0 \\
E(x(i)^2) &\le 2\delta_r(i)^2 \sigma^2 \\
\sum_i E(x(i)^2) &\le 2\sigma^2 \|\delta_r\|^2 \le 162ck\sigma^4 \frac{1}{w_{\min}} \left(1 + \frac{n}{m}\right)
\end{align*}

Also note that $|x(i)| \le |\delta_i| \le \frac{2\sigma^2}{w_{\min}}$. This is simply because the number of 1's in a
column can be at most $1.1m\sigma^2$. Hence $|\delta_i| \le \frac{1}{p_r}1.1m\sigma^2 \le \frac{2\sigma^2}{w_{\min}}$.

We are ready to apply Bernstein's inequality, using which we get,
\[
\Pr\left\{\left|\sum_{i \in [n]} x(i)\right| \ge 15ck\sigma^2 \frac{1}{w_{\min}} \left(1 + \frac{n}{m} + \log m\right)\right\}
\]
\[
\le \exp\left(\frac{-225c^2k^2\sigma^4 \frac{1}{w_{\min}^2}\left(1 + \frac{n}{m} + \log m\right)^2}{162ck\sigma^4\frac{1}{w_{\min}}\left(1 + \frac{n}{m}\right) + 30ck\frac{\sigma^4}{w_{\min}^2}\left(1 + \frac{n}{m} + \log m\right)}\right)
\le \frac{1}{m^4} \qedhere
\]
\end{proof}

The correctness of the algorithm follows:

\begin{theorem}\label{thm:main}
$\textsc{Cluster}(A, k)$ successfully clusters the matrix $A$.
\end{theorem}

\begin{proof}
It suffices to show that $\textsc{Project}(A, \mu_1^*, \ldots, \mu_k^*)$ works. Merging the clusters
from two calls of $\textsc{Project}$ is easy just by comparing the respective centers.

Now, let $u \in T_r$. For all $t \ne r$,
\begin{align*}
&(v - \mu_r^*) \cdot (\mu_t^* - \mu_r^*) + (v - \mu_t^*) \cdot (\mu_r^* - \mu_t^*) \\
&= (\mu_r^* - \mu_t^*) \cdot (v - \mu_t^* - v + \mu_r^*) \\
&= (\mu_r^* - \mu_t^*) \cdot (-\mu_t^* + \mu_r^*) \\
&= \|\mu_r^* - \mu_t^*\|^2
\end{align*}

By Lemma~\ref{lem:centers_close}
\[
\|\mu_r^* - \mu_t^*\|^2 \ge 0.95\|\mu_r - \mu_t\|^2
\]
\[
\Rightarrow \quad (v - \mu_r^*) \cdot (\mu_t^* - \mu_r^*) + (v - \mu_t^*) \cdot (\mu_r^* - \mu_t^*) \ge 0.95\|\mu_r - \mu_t\|^2
\]
\[
\Rightarrow \quad |(v - \mu_r^*) \cdot (\mu_t^* - \mu_r^*)| + |(v - \mu_t^*) \cdot (\mu_r^* - \mu_t^*)| \ge 0.95\|\mu_r - \mu_t\|^2
\]

As $|(v - \mu_r^*) \cdot (\mu_t^* - \mu_r^*)| \le 0.4\|\mu_r - \mu_t\|^2$ by Lemma~\ref{lem:project}
\[
|(v - \mu_t^*) \cdot (\mu_r^* - \mu_t^*)| \ge 0.55\|\mu_r - \mu_t\|^2 > |(v - \mu_r^*) \cdot (\mu_t^* - \mu_r^*)|
\]
\end{proof}
This proves our claim.

As mentioned in the introduction, our analysis extends beyond the case of Bernoulli
distributions. Bernstein's inequality works for subgaussian distributions \cite{petrov75}, and
similar bounds are available for vectors with limited independence as well (e.g.\ \cite{schmidt95}).
Given these bounds, all we need to complete the proof for these cases is a bound on
the spectral norm of $\|A - E(A)\|$ for such distributions, which is also available (e.g.\
\cite{dasgupta07, rudelson99}).


\begin{thebibliography}{99}

\bibitem{achlioptas05}
D.~Achlioptas and F.~McSherry.
On spectral learning of mixtures of distributions.
In \emph{Conference on Learning Theory (COLT)}, pages 458--469, 2005.

\bibitem{alon97}
N.~Alon and N.~Kahale.
A spectral technique for coloring random 3-colorable graphs.
\emph{SIAM J. Comput.}, 26(6):1733--1748, 1997.

\bibitem{alon98}
N.~Alon, M.~Krivelevich, and B.~Sudakov.
Finding a large hidden clique in a random graph.
In \emph{Annual ACM-SIAM Symposium on Discrete Algorithms (SODA)}, pages 594--598, 1998.

\bibitem{boppana87}
R.~Boppana.
Eigenvalues and graph bisection: an average case analysis.
In \emph{IEEE Symposium on Foundations of Computer Science (FOCS)}, pages 280--285, 1987.

\bibitem{dasgupta07}
A.~Dasgupta, J.~Hopcroft, R.~Kannan, and P.~Mitra.
Spectral clustering with limited independence.
In \emph{Annual ACM-SIAM Symposium on Discrete Algorithms (SODA)}, pages 1036--1045, 2007.

\bibitem{dasgupta04}
A.~Dasgupta, J.~Hopcroft, and F.~McSherry.
Spectral analysis of random graphs with skewed degree distributions.
In \emph{IEEE Symposium on Foundations of Computer Science (FOCS)}, pages 602--610, 2004.

\bibitem{drineas99}
P.~Drineas, A.~Frieze, R.~Kannan, S.~Vempala, and V.~Vinay.
Clustering in large graphs and matrices.
In \emph{Annual ACM-SIAM Symposium on Discrete Algorithms (SODA)}, pages 291--299, 1999.

\bibitem{jain01}
K.~Jain and V.~Vazirani.
Approximation algorithms for metric facility location and median problems using the primal-dual schema and lagrangian relaxation.
\emph{J. ACM}, 48(2):274--296, 2001.

\bibitem{kannan05spectral}
R.~Kannan, H.~Salmasian, and S.~Vempala.
The spectral method for general mixture models.
In \emph{Conference on Learning Theory (COLT)}, pages 444--457, 2005.

\bibitem{kanungo04}
T.~Kanungo, D.~Mount, N.~Netanyahu, C.~Piatko, R.~Silverman, and A.~Wu.
A local search approximation algorithm for k-means clustering.
\emph{Comput. Geom.}, 28(2-3):89--112, 2004.

\bibitem{mcsherry01}
F.~McSherry.
Spectral partitioning of random graphs.
In \emph{IEEE Symposium on Foundations of Computer Science (FOCS)}, pages 529--537, 2001.

\bibitem{petrov75}
V.~Petrov.
\emph{Sums of independent random variables}.
Springer, 1975.

\bibitem{rudelson99}
M.~Rudelson.
Random vectors in the isotropic position.
\emph{J. of Functional Analysis}, 168(1):60--72, 1999.

\bibitem{schmidt95}
J.~Schmidt, A.~Siegel, and A.~Srinivasan.
Chernoff-Hoeffding bounds for applications with limited independence.
\emph{SIAM J. Discret. Math.}, 8(2):223--250, 1995.

\bibitem{vempala04}
S.~Vempala and G.~Wang.
A spectral algorithm for learning mixture models.
\emph{Journal of Computer and System Sciences}, 68(4):841--860, 2004.

\bibitem{vu05}
V.~Vu.
Spectral norm of random matrices.
In \emph{ACM Symposium on Theory of computing (STOC)}, pages 619--626, 2005.

\bibitem{mitra08thesis}
P.~Mitra.
\emph{Clustering algorithms for random and pseudo-random structures}.
PhD thesis, Yale University, 2008.

\bibitem{neumann18}
S.~Neumann.
Bipartite stochastic block models with tiny clusters.
In \emph{Advances in Neural Information Processing Systems (NeurIPS)}, 2018.

\bibitem{seifchen24}
M.~Seif and Y.~Chen.
Clustering mixtures of discrete distributions: a note on {M}itra's algorithm.
arXiv preprint arXiv:2405.19559, 2024.

\end{thebibliography}
\end{document}